\documentclass[conference]{IEEEtran}
\usepackage{amsmath,amssymb,amsthm,bm}
\usepackage[colorlinks,
            linkcolor=blue,
            citecolor=blue,
            urlcolor=magenta,
            pdftex
]{hyperref}

\newtheorem{theorem}{Theorem}

\newtheorem{lemma}[theorem]{Lemma}
\newtheorem{remark}[theorem]{Remark}

\ifCLASSINFOpdf
\else
\fi
\begin{document}
%
\title{On the Scaling of Interference Alignment Under Delay and Power Constraints}



%
\author{\IEEEauthorblockN{Subhashini Krishnasamy\IEEEauthorrefmark{1},
Urs Niesen\IEEEauthorrefmark{2}, and
Piyush Gupta\IEEEauthorrefmark{2}}
\IEEEauthorblockA{\IEEEauthorrefmark{1}Dept. of Electrical and Computer Engineering,
The University of Texas at Austin,
Austin, TX, USA\\ Email: \emph{subhashini.kb@utexas.edu}}
\IEEEauthorblockA{\IEEEauthorrefmark{2}Qualcomm New Jersey Research Center,
Bridgewater, NJ, USA\\
Email: \emph{\{urs.niesen, p.gupta\}@ieee.org}}
}


\maketitle

\begin{abstract}
Future wireless standards such as 5G envision dense wireless networks with large number of simultaneously connected devices. In this context, interference management becomes critical in achieving high spectral efficiency. Orthogonal signaling, which limits the number of users utilizing the resource simultaneously, gives a sum-rate that remains constant with increasing number of users. An alternative approach called interference alignment promises a throughput that scales linearly with the number of users. However, this approach requires very high SNR or long time duration for sufficient channel variation, and therefore may not be feasible in real wireless systems. We explore ways to manage interference in large networks with delay and power constraints. Specifically, we devise an interference phase alignment strategy that combines precoding and scheduling without using power control to exploit the diversity inherent in a system with large number of users. We show that this scheme achieves a sum-rate that scales almost logarithmically with the number of users. We also show that no scheme using single symbol phase alignment, which is asymmetric complex signaling restricted to a single complex symbol, can achieve  better than logarithmic scaling of the sum-rate.
\end{abstract}


\begin{IEEEkeywords}
Interference alignment, delay, power constraints, large networks, diversity.
\end{IEEEkeywords}

%
\IEEEpeerreviewmaketitle

\section{Introduction}
\label{sec:intro}

With the rise in popularity of wireless applications, cellular networks are quickly progressing toward ultra-dense deployment of cells. This is one of the central ideas proposed in the development of future wireless standards like 5G \cite{osseiran-etal14vision-5g, baldemair-etal13evolving-wireless-comm}. Moreover, with efforts to introduce machine-type devices and device-to-device communication, the number of users in the network communicating simultaneously is expected to grow significantly \cite{osseiran-etal14vision-5g, baldemair-etal13evolving-wireless-comm}. With spectrum being a limited resource, managing interference becomes critical for good performance in such networks. In addition to high data rates, these next generation applications demand low latency and low power consumption \cite{osseiran-etal14vision-5g}. These conflicting requirements make the task of interference management in large networks challenging. 

In this context, our work focuses on interference management strategies that seek to optimize the total throughput subject to delay, power, spectrum, and device constraints. Specifically, we consider the $K$-user time-invariant Gaussian interference channel in which users have limited power and only single transmit-receive antennas, and we study the scaling behavior of the achievable sum-rate across all users.

A traditional way of dealing with interference is to \emph{avoid} it through the use of orthogonal access schemes that prevent overlap of transmissions from different users in time or frequency. In such schemes, each user is allocated $\frac{1}{K}$ of the total resource. With average-power constraints at the transmitters, a scaling of $\Theta(\ln K)$ can be achieved through bursty power transmission. But such bursty transmission, which requires high peak-power, might not be feasible due to hardware constraints and one is compelled to limit the peak transmit power. With peak-power constraint, the sum-rate achieved by these orthogonal access schemes remains constant with increasing number of users. Another approach to managing interference is to \emph {ignore} it by treating interference as noise. When interference is treated as noise, as in the previous case, the sum-rate again scales as a constant with increasing number of users.

In contrast, various \emph{interference alignment} schemes proposed in recent years promise sum-rates that scale linearly with the number of users. These schemes are based on the idea of aligning the interference from unintended transmitters at the receivers. They rely on the system resources to furnish the \emph{diversity} required to achieve this alignment. For example, the vector interference alignment scheme proposed by Cadambe and Jafar \cite{cadambe-jafar08linear-interf-alignment} is shown to achieve $K/2$ degrees of freedom in a time-varying $K$-user Gaussian interference channel with the use of $K^{\Omega(K^2)}$ independent channel realizations. Ozgur and Tse present in \cite{ozgur-tse09linear-interf-alignment} an alternate version of this scheme for phase-fading channels that achieves a linear scaling of the sum-rate at finite SNR. Their scheme requires the number of channel realizations to grow exponentially in $K^2$, i.e., as $2^{\Omega(K^2)}$. Similarly, ergodic interference alignment \cite{nazer-etal12ergodic-interf-alignment} proposed by Nazer \emph{et al.} requires an average of $K^{\Omega(K^2)}$ independent channel realizations to achieve a linear scaling of the sum-rate at finite SNR. All these schemes therefore incur an extremely large latency, or necessitate extremely large bandwidth if the independent channel realizations are obtained from different frequency sub-bands. For instance, in a 7-user network with channel coherence time of $1$ ms, a requirement of $2^{K^2} = 2^{49}$ independent channel realizations roughly translates to a latency of $18000$ years! Alignment strategies that use lattice codes are based on interference alignment at the signal level and work for time-invariant channels \cite{motahari-etal09real-interf-alignment}. For these lattice schemes, the transmit power (SNR) required to achieve a linear scaling of the sum-rate grows exponentially in $K^2$. Once again considering the example of a 7-user network, to achieve a rate of $\frac{K}{2}\ln(10 dB),$ each user requires a transmit SNR of more than $250dB$! The required diversity can also be obtained by using multiple transmit and receive antennas, and it has been shown in \cite{razaviyayn12dof-mimo-interf-alignment} that the number of antennas is required to scale linearly with $K$ to be able to achieve a linear scaling of the degrees of freedom for a time invariant $K$-user channel. With a limited number of channel realizations, limited transmit power, and a fixed number of transmit-receive antennas, the sum-rate for all these interference alignment schemes scales as $o(1)$ with increasing number of users.

To summarize, on the one hand, the orthogonal signaling schemes operate with limited power requirements and latency but, with peak-power constraint, can achieve only a constant sum-rate with increasing number of users. On the other hand, currently known interference alignment schemes achieve a linear scaling of the sum-rate but require prohibitively large delay or power to achieve this scaling. This motivates us to ask if one can achieve a better scaling in the presence of delay and transmit power constraints. 

While existing results provide insight into interference management in small networks with a fixed number of users and with system resources tending to infinity, we are interested in the scenario where the system resources (transmit power and number of independent channel realizations) are fixed and the number of users approaches infinity. \cite{bresler-tse09diversity-interf-channel, li-ozgur14diversity-interf-alignment} study an intermediate problem of characterizing the degrees of freedom as a function of the number of independent channel realizations. In \cite{bresler-tse09diversity-interf-channel}, the authors give the exact expression for the degrees of freedom as a function of the number of independent channel realizations when restricted to vector interference alignment strategies for a $3$-user Gaussian interference channel. \cite{li-ozgur14diversity-interf-alignment} extends this result by providing upper bounds for the degrees of freedom for the $K$-user case.

We study the problem of characterizing the sum-rate for a fixed transmit power and a single realization of the $K$-user complex Gaussian interference channel. Specifically, we focus on the scaling of achievable sum-rates with increasing number of users in the network. Scaling laws for the capacity of large networks have been studied before in the context of multi-hop networks \cite{gupta-kumar00capacity-wireless-nws, ozgur07hierarchical-coop}. Capacity scaling for time-varying phase-fading interference network \emph{without} delay constraints is analyzed in \cite{jafar09capacity-phase-fading-interf-nws}, and it is shown that, with increasing number of users, the capacity converges in probability to the rate achieved by ergodic interference alignment. The high-SNR performance of opportunistic interference alignment in large networks, and again, without delay constraints was studied in \cite{tajer-wang12nk-user}.

The main contributions of this work are threefold. First, we formulate the problem of capacity scaling with increasing number of users for an interference channel that reflects delay and power constraints in the network. Specifically, we study the $K$-user time-invariant complex Gaussian interference channel where users are subject to peak transmit power constraint. The time-invariant channel captures users' delay constraints while the peak-power constraint models practical power limitations in a network. Such a network lacks the diversity that can be derived from a time-varying channel or high SNR. The central question we address in this paper is the feasibility of simultaneously aligning interference at multiple receivers with increasing number of users.

Second, for this channel model, we propose an achievable scheme that can obtain a sum-rate of $\Omega \left(\frac{\ln K}{\ln \ln K}\right) \; w.h.p.$ This is achieved by clustering users into groups such that within each group the interference signals align in phase. The clusters are then scheduled in a round-robin fashion. Such interference \emph{phase alignment} is made possible by the intrinsic diversity in large networks. 

Finally, we show that, within the class of schemes called \emph{single-symbol phase alignment}, the best achievable sum-rate is $O(\ln K) \; w.h.p.$ This class includes strategies that use different transmit and receive directions in the single-symbol complex plane while using real-valued Gaussian codebooks. This method of using real-valued Gaussian codebooks and optimizing for transmit directions is introduced in \cite{cadambe-etal10asymmetric-complex-signaling} as \emph{asymmetric complex signaling.} The authors demonstrate the benefit of designing the transmit directions by proving the achievability of 1.2 degrees of freedom in a $K$-user time-invariant interference channel for almost all values of channel coefficients. This is higher than the previously conjectured upper bound of 1 degree of freedom \cite{host-madsen-nosratinia05multiplexing-gain}. The converse result in our work for single-symbol phase alignment schemes shows that designing the transmit directions over a single symbol cannot provide better than $O(\ln K)$ scaling of the sum-rate. In contrast, the proposed phase alignment scheme achieves \emph{almost} $\Omega(\ln K)$ scaling through scheduling without power control simply by exploiting the diversity inherent in large networks.

The remainder of the paper is organized as follows. Section~\ref{sec:problem} formally introduces the problem of capacity scaling under delay and power constraints. The main results are presented in Section~\ref{sec:main} -- the phase alignment scheme that achieves a sum-rate of $\Omega \left(\frac{\ln K}{\ln \ln K}\right)$ is given in Section~\ref{subsec:achievability}, while the upper bound of $O(\ln K)$ for the class of single-symbol phase alignment schemes is given in Section~\ref{subsec:converse}. We present detailed proofs of the results in the Appendix.

\section{Problem Setting}
\label{sec:problem}
We consider a $K$-user complex Gaussian interference channel. There are $K$ transmitter-receiver pairs (indexed from $1$ to $K$) with each transmitter desiring to communicate its message to the corresponding receiver. To focus on the gains of interference alignment, we assume a phase-fading channel (see Remark~\ref{rem:model}) -- the signal from Transmitter $j$ is received at Receiver $k$ with a phase rotation of $\Theta_{kj}$. The $\{\Theta_{kj}\}$ are drawn independently according to a uniform distribution in $[-\pi, \pi)$. These channel gains are time invariant and hence offer no diversity across time. Thus, the channel is represented by the following input-output relationship: $\forall \, k \in [K], \, t \in \mathbb{N},$
\begin{align*}
Y_k[t] = \sum_{j} e^{i\Theta_{kj}} X_{j}[t] + Z_k[t],
\end{align*}
where $Y_k[t]$ is the received signal at Receiver $k$ and $X_j[t]$ is the signal transmitted by Transmitter $j$ in time slot $t$. $Z_k[t]$ is the white additive noise at Receiver $k$ and is assumed to be drawn from circularly symmetric complex Gaussian distribution ($\mathcal{CN}(0,1)$). We assume that all transmitters and receivers have prior knowledge of the channel phases.

For a code of block length $n$, Transmitter $k$ wishes to communicate message $W_k$ distributed uniformly in $\{1, 2, \dots, 2^{nR_k}\}$ to Receiver $k$. Transmitter $k$ uses encoder, $f_k(n): W_k \mapsto X_k^n$ which maps the message $W_k$ to the transmit signal $X_k^n = (X_k[1], X_k[2], \dots, X_k[n]) \in \mathbb{C}^n$. The transmit signals are subject to unit per-symbol peak-power constraint, i.e., $ \lvert X_{k}[t] \rvert ^2 \leq 1, \; \forall k \in [K], \: \forall t \in \mathbb{N}$. The decoding function at Receiver $k$, $\phi_k(n): Y_k^n \mapsto \hat{W}_k$ maps the received signal $Y_k^n = (Y_k[1], Y_k[2], \dots, Y_k[n])$ to an estimate of the transmitted message, $\hat{W}_k \in \{1, 2, \dots, 2^{nR_k}\}$. We use the standard definition for achievable rates -- a rate vector $(R_1, R_2, \dots R_K)$ is achievable if there is a sequence of encoding and decoding functions ($f_1(n), f_2(n), \dots f_K(n); \phi_1(n), \phi_2(n), \dots \phi_K(n)$) such that the average probabilities of error, $\mathbb{P}[\hat{W}_k \neq W_k], \, k \in [K]$ all go to zero as $n$ goes to infinity. A sum-rate, $R_{sum}$ is achievable if there exists an achievable rate vector $(R_1, R_2, \dots R_K)$ such that $R_{sum} = \sum_{k=1}^K R_k.$ We are interested in the asymptotics of the sum-rate with increasing number of users.

\paragraph*{Notation} Throughout the paper, we use the term \emph{users/nodes} to refer to transmitter-receiver pairs. 
We say that an event occurs \emph{with high probability} $(w.h.p.)$ if its probability goes to $1$ as the number of users, $K$ goes to infinity. The probability is with respect to the randomness in the channel coefficients, i.e., $\{\Theta_{kj}\}$. 

\section{Main Results}
\label{sec:main}
In this section, we present our main results, which provide an asymptotic lower bound (Theorem~\ref{thm:achievability}) and upper bound (Theorem~\ref{thm:upper-bound}) for the achievable sum-rate as the number of users in the network $K$ increases.
\begin{theorem}
\label{thm:achievability}
There exists a scheme that achieves a sum-rate of $\Omega \left(\frac{\ln K}{\ln \ln K}\right) \; w.h.p.$ as $K \rightarrow \infty.$
\end{theorem}
The upper bound is for the class of schemes called single-symbol phase alignment described in detail in Section~\ref{subsec:converse}.
\begin{theorem}
\label{thm:upper-bound}
No single-symbol phase alignment scheme can achieve a sum-rate better than $O(\ln K) \; w.h.p.$ as $K \rightarrow \infty.$
\end{theorem}
\subsection{Achievability of $\Omega \left(\frac{\ln K}{\ln \ln K}\right)$}
\label{subsec:achievability}
The signaling scheme achieving the sum-rate scaling in Theorem~\ref{thm:achievability} is a \emph{phase alignment} strategy that combines precoding and scheduling to utilize the diversity in the channel gains across users. In each time slot $t$, a subset of users, $\mathcal{S}[t] \subset [K]$ is scheduled for transmission. Users not in this subset do not transmit. Before we describe the scheduling strategy, we first describe the precoding scheme for the scheduled users. The precoding scheme is used to transform the complex-valued phase-fading Gaussian channel into a real-valued Gaussian channel. This precoding scheme is an instance of \emph{asymmetric complex signaling} introduced in \cite{cadambe-etal10asymmetric-complex-signaling}. The main idea in asymmetric complex signaling is to treat the complex channel input as having two real dimensions which are used for the design of transmit directions for interference alignment. The transmit direction along which the data stream is transmitted can be extended to a multi-dimensional real vector through the use of symbol extensions.
\subsubsection{\textbf{Precoding}}
In our achievable scheme, we use the two-dimensional real vector space without symbol extensions for designing the transmit directions.
\begin{itemize}
\item Encoder:
The data stream of a scheduled user $k,$ $\{U_k(t)\}$ $(U_k(t) \in \mathbb{R}),$ is transmitted along the direction $e^{i \alpha_k}$ at maximum power, where $\alpha_k = -\Theta_{kk}$. Therefore, for any scheduled user $k \in \mathcal{S}[t],$ the transmit signal is given by $$X_k[t] =  U_k[t] e^{-i \Theta_{kk}}.$$
\item Decoder:
The receivers corresponding to the scheduled users project their received signal along direction $e^{i\gamma_k}$ to decode their respective symbols. We choose $\gamma_k = 0$ for all $k.$ Therefore, for any $k \in \mathcal{S}[t],$ the received signal and the noise after projection are respectively given by $$\tilde{Y}_k[t] = \mathsf{Re}(Y_k[t]), \; \tilde{Z}_k[t] = \mathsf{Re}(Z_k[t]).$$ 
\end{itemize}
Thus, the encoder and decoder generate a channel in the real domain with effective input-output relationship,
\begin{align*}
\tilde{Y}_k[t] = U_k[t] + \sum_{\substack{j \in \mathcal{S}[t] \\ j \neq k}} \cos \left(\Theta_{kj} - \Theta_{jj}\right) U_{j}[t] + \tilde{Z}_k[t].
\end{align*}
 Since $Z_k[t] \sim \mathcal{CN}(0,1),$ we have $\tilde{Z}_k[t] \sim \mathcal{N}(0,\frac{1}{2})$.

\subsubsection{\textbf{Phase Alignment through Scheduling}}
We now describe the scheduling strategy that determines the subset of users transmitting in each time slot. While choosing small subsets imposes a fundamental limit on the sum-rate, choosing large subsets could result in high interference which again adversely impacts the achievable sum-rate. Ideally, we would like to schedule large subsets of users such that within each subset, the effective interference is low. In other words, at each receiver in the scheduled subset, the interference from other users in the subset are phase aligned.

\textbf{Interference Graph.} To this end, we construct an \emph{Interference Graph} $G(V, E)$ with the users comprising the node set $V$. An edge exists between two nodes $k$ and $j$ if one of the effective cross gains, $ \lvert \cos \left(\Theta_{kj} - \Theta_{jj}\right) \rvert $ or $ \lvert \cos \left(\Theta_{jk} - \Theta_{kk}\right) \rvert $ is more than $\frac{\pi}{\sqrt{\ln K}}$.  Users $k$ and $j$ are said to \emph{interfere} with each other if there is an edge between them.

\textbf{Graph Coloring.} As a part of the scheduling strategy, the set of users is first partitioned into subsets of nodes that do not interfere with each other, i.e., into independent subsets. A greedy graph coloring algorithm is used to find such a partition. 
The greedy algorithm can be described as follows -- color each node sequentially (from $1$ to $K$) with the first color not used by any of its already colored neighbors. 

Given the partition obtained from the coloring algorithm, each independent subset, corresponding to the set of nodes with the same color, is scheduled in a round-robin fashion. 

We now give a sketch of the proof of Theorem~\ref{thm:achievability} which shows that the proposed precoding-scheduling scheme achieves a sum-rate of $\Omega \left(\frac{\ln K}{\ln \ln K}\right) \; w.h.p.$ Full details of the proof are given in the Appendix. 
We first show (Lemma~\ref{lem:greedy-alg} in the Appendix) that the greedy graph coloring algorithm partitions the nodes into $O(\frac{K\ln \ln K}{\ln K})$ independent subsets $ w.h.p.$ We then show that scheduling independent subsets limits the strength of the interference signal, thereby giving each user a strictly positive rate across the scheduled time slots. This result in combination with the above lemma proves Theorem~\ref{thm:achievability}.
\qed

From Theorem~\ref{thm:achievability}, we see that phase alignment through a combination of precoding and scheduling provides the benefit of diversity in a large network. The precoding part of the scheme is used to transform the complex-valued phase-fading channel gains into real-valued channel gains such that the magnitudes of the effective cross gains in the real channel exhibit sufficient diversity across users. The scheduling part of the scheme exploits this diversity by partitioning the users into subsets such that the interference is aligned at all users within each subset. 

This method of utilization of network diversity to achieve high rates parallels other interference alignment schemes that attempt to align the interference at the receivers. Vector interference alignment \cite{cadambe-jafar08linear-interf-alignment, ozgur-tse09linear-interf-alignment} and ergodic interference alignment \cite{nazer-etal12ergodic-interf-alignment} achieve a linear scaling of the sum-rate by using the diversity of the channel gains across time provided by the time varying channel. Analogously, lattice coding schemes \cite{motahari-etal09real-interf-alignment}, which achieve alignment at the signal level, rely on the distinct scaling of transmit signals possible at high transmit power. But these alignment schemes are practical only in small networks since the diversity required to align the interfering signals at all the receivers increases very quickly with the number of users. For instance, to achieve a linear scaling of the sum-rate, vector and ergodic interference alignment schemes require $2^{\Omega(K^2)}$ independent channel realizations. Similarly, the SNR required to achieve signal level alignment for linear scaling of sum-rate in lattice codes scales as  $2^{\Omega(K^2)}.$

When the number of independent channel realizations and the transmit power are limited (fixed with respect to the number of users), the above schemes can only provide a sum-rate that scales as $o(1)$ with increasing number of users. In the absence of diversity from signal level and time variations, the proposed phase alignment strategy utilizes the diversity due to independent channel realizations across large number of users. The key idea here is that, although it is difficult to align the interference at all the $K$ receivers simultaneously, it is possible to schedule users such that the effective interference is aligned in phase for all the scheduled users in every time slot. As Theorem~\ref{thm:achievability} shows, unlike the orthogonal signaling schemes which give only $O(1)$ sum-rate with peak-power constraint, the phase alignment strategy achieves \emph{almost} $O(\ln K)$ sum-rate for most channel realizations. A notable aspect of this scheme is that it does not require power control. 
 Moreover, the scheme ensures fairness among users by scheduling all users as opposed to scheduling only a partial subset of users. This facilitates the extension of the above result 
 to the scaling of achievable symmetric rate.
 
\begin{remark} 
\label{rem:model}
Similar achievability results can be obtained for channel models other than those considered in this paper. For example, a sum-rate of $\Theta(\ln K)$ can be achieved through bursty power transmission if the peak-power constraints are relaxed to average-power constraints. Similarly, for a Rayleigh fading channel, $\Theta(\ln K)$ scaling can be achieved by allowing only the user with the largest magnitude of channel gain to transmit. But these schemes cannot be applied under the model considered in this paper with peak-power constraints and phase fading channel. Therefore, this model is effective in bringing into sharp focus the gains in sum-rate that can be achieved through opportunistic alignment strategies.
\end{remark}

\subsection{Upper Bound of $O(\ln K)$}
\label{subsec:converse}
 For the upper bound, we restrict our attention to the class of signaling schemes described below.
 
\textbf{Single-Symbol Phase Alignment.} A signaling scheme is a single-symbol phase alignment scheme if
\begin{enumerate}
\item the encoders use real Gaussian codebooks with unit average-power constraint and linear precoding (design of transmit directions) limited to a single time slot,
\item the decoders treat interference as noise.
\end{enumerate}
The single-symbol phase alignment class consists of all possible asymmetric complex signaling schemes \cite{cadambe-etal10asymmetric-complex-signaling} with Gaussian codebooks and transmit directions restricted to single complex symbol in combination with power control. Note that the the peak-power constraint in Section~\ref{sec:problem} has been relaxed to an average-power constraint. The phase alignment through scheduling scheme described in Section~\ref{subsec:achievability} falls under this class if the same subset of users are scheduled in all time slots and the data stream for all these scheduled users are chosen to be Gaussian with unit average-power. Thus, an argument similar to that in Theorem~\ref{thm:achievability} shows that a sum-rate of $\Omega \left(\frac{\ln K}{\ln \ln K}\right)$ is achievable with single-symbol phase alignment schemes. Theorem~\ref{thm:upper-bound} conveys that the maximum achievable sum-rate in this class is $O(\ln K)$ for most channel realizations.


The result can be mathematically described as follows: \\
Assume wlog that transmitter $k$ uses transmit direction $e^{i(\alpha_k - \Theta_{kk})}$ and average-power $P_k$ ($\in [0,1]$) to transmit its data stream $\{U_k\}$ from a real Gaussian codebook. Thus the received signal at Receiver $k$ can be written as
\begin{align*}
Y_k[t] & = \sqrt{P_k} e^{i\alpha_k} U_k[t] + \sum_{j\neq k}\sqrt{P_j} e^{i(\alpha_j + \Theta_{kj} - \Theta_{jj})} U_j[t] + Z_k[t].
\end{align*}
Since the decoders treat interference as noise, a linear filter is optimal for Gaussian input. Let Receiver $k$ project its received signal along the direction $e^{i\gamma_k}$. The output of this filter is given by
\begin{align*}
\tilde{Y}_k[t] {}={} & \sqrt{P_k} \cos(\alpha_k - \gamma_k ) U_k[t]	\\
& + \sum_{j\neq k}\sqrt{P_j} \cos(\alpha_j + \Theta_{kj} - \Theta_{jj} - \gamma_k ) U_j[t] + \tilde{Z}_k[t],
\end{align*}
where $\tilde{Z}_k[t] \sim \mathcal{N}(0,\frac{1}{2})$.
Define $$\beta_{kj} := \cos^2(\alpha_j + \Theta_{kj} - \Theta_{jj} - \gamma_k) \, \forall j, k.$$ Then user $k$ obtains a rate of
\begin{align*}
R_k(\mathbf{P}, \boldsymbol{\alpha}, \boldsymbol{\gamma}) & = \ln\left(1 + \frac{P_k \beta_{kk}}{\sum_{j\neq k}P_j \beta_{kj} + \frac{1}{2}}\right).
\end{align*}
Therefore, the sum-rate achieved by a scheme with power allocation $\mathbf{P}$, transmitter precoding $\boldsymbol{\alpha}$, and receiver precoding $\boldsymbol{\gamma}$ is given by
\begin{align*}
R_{sum}(\mathbf{P}, \boldsymbol{\alpha}, \boldsymbol{\gamma}) & = \sum_{k=1}^K R_k(\mathbf{P}, \boldsymbol{\alpha}, \boldsymbol{\gamma}),
\end{align*}
and the best sum-rate achievable in this class of schemes is
\begin{align*}
\max_{\substack{\mathbf{P} \in [0,1]^K\\ (\boldsymbol{\alpha}, \boldsymbol{\gamma}) \in [-\pi, \pi)^{2K}}} \sum_{k=1}^K \ln\left(1 + \frac{P_k \beta_{kk}}{\sum_{j\neq k}P_j \beta_{kj} + \frac{1}{2}}\right). 
\end{align*}

Theorem~\ref{thm:upper-bound} shows that the above quantity is $O(\ln K) \; w.h.p.$ Again, we give here a sketch of the proof of Theorem~\ref{thm:upper-bound}. The complete proof is provided in the Appendix. 
We first show that, for asymptotic analysis, it is sufficient to consider a restricted class of scheduling strategies that do not use power control but schedule a subset of users who transmit at maximum power (Lemma~\ref{lem:power-ctrl=schedule} in the Appendix). Note that this restricted class of scheduling strategies is a subset of the original class of transmission schemes with powers $P_k$ restricted to either $0$ or $1$. We then derive a probabilistic upper bound on the sum-rate for any fixed scheduling strategy by fixing the set of users who transmit as well as their transmit and receive directions (Lemma~\ref{lem:fixed-parameters} in the Appendix). We then show that, for any fixed set of scheduled users, a slight perturbation of the transmit and receive directions does not change the sum-rate by a large amount (Lemma~\ref{lem:continuity} in the Appendix). This continuity property enables us to extend the probabilistic upper bound on the achievable sum-rate to the class of scheduling strategies with a fixed set of scheduled users with transmit and receive directions within a generic small set. A union bound over sets that cover the space of all possible transmit and receive directions gives an upper bound for the maximum achievable sum-rate for a fixed set of scheduled users. To show that no scheduling strategy can achieve a better scaling than $O(\ln K)$ $w.h.p.$ we take a union bound over all sets of scheduled users. 
\qed

The example of asymmetric complex signaling in \cite{cadambe-etal10asymmetric-complex-signaling} for a $3$-user Gaussian interference network suggests the potential benefits of optimizing the transmit directions for interference alignment. The authors show that a careful design of the transmit directions over $5$ symbol extensions can achieve $1.2$ degrees of freedom as opposed to the upper bound of $1$ degree of freedom conjectured in \cite{host-madsen-nosratinia05multiplexing-gain} for time-invariant channels. They also prove that $1.2$ degrees of freedom is the maximum achievable for the $3$-user case. This upper bound is shown to be due to the impossibility of aligning the interference signal from a transmitter at more than one unintended receiver. As mentioned in \cite{cadambe-etal10asymmetric-complex-signaling}, it is interesting to understand how this fundamental limitation affects achievable rates in the general $K$-user interference channel.

Theorem~\ref{thm:upper-bound} answers this question partially from a scaling perspective when the design of the transmit directions is limited to a single-symbol extension. As seen in the proof of Theorem~\ref{thm:upper-bound} (specifically Lemma~\ref{lem:fixed-parameters}), the conflicting goals of aligning the interference signal at different receivers limits the achievable sum-rate to $O(\ln K).$ Unlike the linear algebraic proof techniques generally used to prove upper bounds for vector interference alignment schemes (as in \cite{cadambe-etal10asymmetric-complex-signaling} to prove the upper bound for the $3$-user case), our proof bears resemblance to Khintchine's method of proving upper bounds for Diophantine approximation (for example, see \cite{dodson09diophantine-khintchine}). 

\section{Conclusion}
\label{sec:conclusion}
We considered the fully connected interference network under fixed power and delay constraints. For time-invariant interference networks, the idea of complex asymmetric signaling is one step in the direction of achieving better rates than orthogonal signaling schemes. We proposed a signaling scheme that uses opportunistic user-scheduling as a phase alignment strategy. It is shown that unlike orthogonal signaling and well-known interference alignment strategies that can achieve only a constant scaling of sum-rate with increasing number of users, the proposed scheme achieves a sum-rate that increases with the number of users in the system. The main idea underlying this interference management strategy is the utilization of the diversity natural in large networks. The proposed phase alignment scheme when modified to schedule a single subset of users across all time slots falls under the class of complex asymmetric signaling schemes. 
We show that this class of asymmetric complex signaling schemes restricted to a single complex symbol cannot achieve a much better scaling of sum-rate as compared to the proposed phase alignment scheme.





%

\bibliography{InterfAlignment}


\appendix[Proofs]
\label{sec:proofs}
In this section, we present the complete proofs of Theorem~\ref{thm:achievability} and Theorem~\ref{thm:upper-bound}. To prove the results, we can assume, without loss of generality, that $\Theta_{kk} = 0 \; \forall k$ since $\{\Theta_{kj}\}$ are i.i.d.\ random variables with uniform distribution in $[-\pi, \pi).$

\subsection{Proof of Achievability (Theorem~\ref{thm:achievability})}
The following lemma gives an upper bound on the number of colors required to color the interference graph $G$ defined in Section~\ref{subsec:achievability}.
\begin{lemma}
\label{lem:greedy-alg}
The greedy algorithm colors the graph $G$ with at most $\frac{K \ln \ln K}{\ln K - 3 \ln \ln K} + 1$ colors $w.h.p.$
\end{lemma}

\begin{proof}
We first note that the interference graph is an Erd\"{o}s-R\'{e}nyi random graph $G(K,p),$ where $p$ is the probability that there is an edge between any two nodes. There is no edge between nodes $k$ and $j$ if and only if $\cos^2\Theta_{kj} \leq \frac{\pi^2}{\ln K}$ and $\cos^2\Theta_{jk} \leq \frac{\pi^2}{\ln K}.$ We now derive bounds for the edge probability $p.$ Since $\{\Theta_{kj}\}_{j \neq k}$ are  i.i.d.\ $\sim Unif[-\pi, \pi),$
\begin{align}
1 - p & = \left(\mathbb{P}\left[\cos^2\Theta_{kj} \leq \frac{\pi^2}{\ln K}\right]\right)^2	\nonumber	\\
 & = \left(\mathbb{P}\left[\sin^2\Theta_{kj} \leq \frac{\pi^2}{\ln K}\right]\right)^2	\nonumber	\\
 & \geq \left(\mathbb{P}\left[\Theta^2_{kj} \leq \frac{\pi^2}{\ln K}\right]\right)^2	\nonumber	\\
 & = \frac{1}{\ln K}. \label{eqn:edge-prob-ub}
\end{align}
The following lower bound for $p$ (used in the proof of Theorem~\ref{thm:achievability}) can be shown in a similar fashion.
\begin{align}
1 - p  & = \left(\mathbb{P}\left[\sin^2\Theta_{kj} \leq \frac{\pi^2}{\ln K}\right]\right)^2	\nonumber	\\
 & \leq \left(\mathbb{P}\left[ \min \lbrace \lvert \Theta_{kj} \rvert, \pi - \lvert \Theta_{kj} \rvert \rbrace  \leq \frac{2\pi}{\sqrt{\ln K}}\right]\right)^2	\nonumber	\\
 & = \frac{16}{\ln K}. \label{eqn:edge-prob-lb}
\end{align}
The remainder of the proof follows a similar structure as in \cite{grimmett-mcdiarmid75colouring-random-graphs}. Let
\begin{align}
\label{eqn:num-colors}
L = \left\lceil \frac{K \ln \ln K}{\ln K - 3 \ln \ln K} \right\rceil,
\end{align}
and let $A_k$ be the event that node $k$ is the first to get color $L+1$. If $k$ is the first node to get color $L+1$, then all the preceding $k-1$ nodes were colored using exactly $L$ colors. Let $(C_1, C_2, \dots C_L)$ denote the coloring of the first $k-1$ nodes, i.e., for any $j,l$ such that $1 \leq j \leq k-1$ and $1 \leq l \leq L$, if node $j$ gets color $l$, then $j \in C_l$. For node $k$ to get color $L+1$, it must have at least one neighbor in each of the color classes. Therefore,
\begin{align*}
\mathbb{P}\left[A_k \vert (C_1, C_2, \dots C_L)\right]
 & = \prod_{l=1}^L \left( 1 - (1-p)^{ \lvert C_l \rvert } \right)	\\
& \leq \prod_{l=1}^L \left( 1 - \left(\frac{1}{\ln K}\right)^{ \lvert C_l \rvert } \right)	\\
 & \leq \left( 1 - \left(\frac{1}{\ln K}\right)^{\sum_{l=1}^L \lvert C_l \rvert /L} \right)^L	\\
 & < \left( 1 - \left(\frac{1}{\ln K}\right)^{K/L} \right)^L	\\
 & < \exp\left( -L \left(\frac{1}{\ln K}\right)^{K/L} \right),
\end{align*}
where the first inequality is due to \eqref{eqn:edge-prob-ub} and the second inequality follows from Jensen's inequality applied to the concave function $f(x) = \ln(1-a^x)$ for $a \in [0,1)$ and $x > 0.$
Using the value of $L$ from \eqref{eqn:num-colors},
\begin{align*}
\ln\left( L \left(\frac{1}{\ln K}\right)^{K/L} \right) & = \ln L - \frac{K}{L} \ln\ln K	\\
& \geq \ln K - \ln\ln K + 3 \ln\ln K - \ln K	\\
& = 2 \ln\ln K,
\end{align*}
which implies
\begin{align*}
 \exp\left( -L \left(\frac{1}{\ln K}\right)^{K/L} \right) = o\left(\frac{1}{K}\right).
\end{align*}
It follows that $\mathbb{P}[A_k] = o(\frac{1}{K}) \; \forall k.$ 

Note that $\cup_{k=1}^K A_k$ is the event that the greedy algorithm requires more than $L$ colors. Since $$\mathbb{P}[\cup_{k=1}^K A_k] \leq \sum_{k=1}^K \mathbb{P}[A_k] = o(1),$$ we conclude that the greedy algorithm colors the graph in at most $L$ colors $w.h.p.$ 
\end{proof}
We now prove Theorem~\ref{thm:achievability}.
\begin{proof}[Proof of Theorem~\ref{thm:achievability}]
It follows from Lemma~\ref{lem:greedy-alg} that the greedy algorithm partitions the nodes into $L = O\left(\frac{K\ln \ln K}{\ln K}\right)$ independent subsets. Let $\alpha(G)$ be the size of the maximum independent set in the graph $G.$ It is known (\cite{bollobas-erdos76cliques-random-graphs, matula76clique-size-random-graph}) that, if $p > K^{-\epsilon}$ for every $\epsilon > 0,$ the size of the largest independent set in a $G(K,p)$ has the following upper bound.
 $$\alpha\left(G(K,p)\right) \leq \frac{2\ln K}{\ln\left(\frac{1}{1-p}\right)} \; w.h.p.$$ 
 Since we have $1-p \leq \frac{16}{\ln K}$ from \eqref{eqn:edge-prob-lb}, this result implies that 
 $$\alpha(G) = O\left(\frac{\ln K}{\ln\ln K}\right) \; w.h.p.$$
 This implies that each color class (independent subset) has $O\left(\frac{\ln K}{\ln\ln K}\right)$ nodes.
 
  Given the partition, each independent subset is scheduled in a round-robin fashion. Since the effective cross gain between any pair of nodes in an independent subset is not more than $\frac{\pi}{\sqrt{\ln K}},$ in every time slot, the power of the interference signal for any scheduled node is bounded from above by $\frac{\pi^2 \alpha(G)}{\ln K} = O\left(\frac{1}{\ln\ln K}\right)$ $w.h.p.$ 
Given  that the total interference power at any node is $O\left(\frac{1}{\ln 	\ln K}\right),$ a simple binary coding scheme can be used to achieve a positive rate 
across the scheduled time slots. Since each node is scheduled at least once in every $L$ time slots, the rate achieved by any user is $\Omega\left(\frac{\ln K}{K\ln \ln K}\right),$ which implies that the sum-rate achieved is $\Omega \left(\frac{\ln K}{\ln \ln K}\right) \; w.h.p.$ 
\end{proof}

\subsection{Proof of Upper Bound (Theorem~\ref{thm:upper-bound})}
The following lemma shows that, in order to prove Theorem~\ref{thm:upper-bound}, it is sufficient to consider the restricted subclass of single-symbol phase alignment schemes in which a subset of users are scheduled for transmission and the scheduled users transmit at the maximum power.
\begin{lemma}
\label{lem:power-ctrl=schedule}
For any $\mathbf{P} \in [0,1]^K, (\boldsymbol{\alpha}, \boldsymbol{\gamma}) \in [-\pi, \pi)^{2K}$, 
\begin{align*}
R_{sum}(\mathbf{P}, \boldsymbol{\alpha}, \boldsymbol{\gamma}) & \leq 2\max_{\mathcal{S} \subset [K]} R_{sum}(\mathbf{1}_{\mathcal{S}}, \boldsymbol{\alpha}, \boldsymbol{\gamma}),
\end{align*}
 where $\mathbf{1}_{\mathcal{S}}$ is the $K$-length vector with the $k^{th}$ component equal to $1$ if $k \in \mathcal{S}$ and $0$ otherwise.
\end{lemma}
\begin{proof}
Using the inequality $\ln(1+x) \leq x \; \forall x \geq 0,$ we have
\begin{align*}
R_{sum}(\mathbf{P}, \boldsymbol{\alpha}, \boldsymbol{\gamma})
 & = \sum_{k=1}^K \ln\left(1 + \frac{P_k \beta_{kk}}{\sum_{j\neq k}P_j \beta_{kj} + \frac{1}{2}}\right) 	\\
& \leq \sum_{k=1}^K  \frac{P_k \beta_{kk}}{\sum_{j\neq k}P_j \beta_{kj} + \frac{1}{2}}.
\end{align*}
Now, let $$\Phi(\mathbf{P}) := \sum_{k=1}^K  \frac{P_k \beta_{kk}}{\sum_{j\neq k}P_j \beta_{kj} + \frac{1}{2}}.$$ For any fixed $k$ and fixed $(P_j)_{j \neq k},$ $\Phi$ is a convex function of $P_k$.\footnote{Note that $\Phi$ is \emph{not} a convex function of the \emph{vector} $\mathbf{P}.$} This can be verified by observing that the $k^{th}$ term in the sum is a linear function of $P_k$ and all other terms are convex functions of $P_k.$ Since the maximum value of a convex function on a compact and convex set is achieved at an extreme point, maximizing over all $P_1 \in [0,1]$ for any fixed $(P_j)_{j \neq 1}$ gives the following upper bound for $\Phi$:
\begin{align*}
\Phi(\mathbf{P}) & \leq \max_{\tilde{P}_1 \in \{0,1\}} \Biggl( \frac{\tilde{P}_1 \beta_{11}}{\sum_{j\neq 1}P_j \beta_{1j} + \frac{1}{2}}	\\
&\phantom{ \leq \max_{\tilde{P}_1 \in \{0,1\}} } \quad + \sum_{k=2}^K  \frac{P_k \beta_{kk}}{\tilde{P}_1 \beta_{k1} + \sum_{j\neq 1,k}P_j \beta_{kj} + \frac{1}{2}} \Biggr).
\end{align*}
By successively maximizing over $P_k$ for $k = 2, 3, \dots, K$ with the other variables fixed, we get
\begin{align*}
\Phi(\mathbf{P}) & \leq \max_{\tilde{\mathbf{P}} \in \{0,1\}^K} \sum_{k=1}^K  \frac{\tilde{P}_k \beta_{kk}}{\sum_{j\neq k}\tilde{P}_j \beta_{kj} + \frac{1}{2}}	\\
& \leq 2\max_{\tilde{\mathbf{P}} \in \{0,1\}^K} \sum_{k=1}^K  \ln\left(1 + \frac{\tilde{P}_k \beta_{kk}}{\sum_{j\neq k}\tilde{P}_j \beta_{kj} + \frac{1}{2}}\right) 	\\
& \leq 2\max_{\tilde{\mathbf{P}} \in \{0,1\}^K} R_{sum}(\tilde{\mathbf{P}}, \boldsymbol{\alpha}, \boldsymbol{\gamma})	\\
& = 2\max_{\mathcal{S} \subset [K]} R_{sum}(\mathbf{1}_{\mathcal{S}}, \boldsymbol{\alpha}, \boldsymbol{\gamma}),
\end{align*}
where the second inequality follows from the fact that $x \leq 2\ln(1+x) \; \forall x \in [0,2].$ We can therefore conclude that
\begin{align*}
R_{sum}(\mathbf{P}, \boldsymbol{\alpha}, \boldsymbol{\gamma}) & \leq 2\max_{\mathcal{S} \subset [K]} R_{sum}(\mathbf{1}_{\mathcal{S}}, \boldsymbol{\alpha}, \boldsymbol{\gamma}).
\end{align*}
\end{proof}

We now show that, for a fixed subset of scheduled users $\mathcal{S} \subset [K]$ and for a fixed set of transmit and receive directions $(\boldsymbol{\alpha}, \boldsymbol{\gamma})$, the tail probability for the sum-rate decays at least exponentially with rate of decay proportional to the size of the scheduled set.
\begin{lemma}
\label{lem:fixed-parameters}
For any fixed $\mathcal{S} \subset [K]$ with $|\mathcal{S}| = s$ and fixed $(\boldsymbol{\alpha}, \boldsymbol{\gamma}) \in [-\pi, \pi)^{2K}$, $$\mathbb{P}\left[R_{sum}(\mathbf{1}_{\mathcal{S}}, \boldsymbol{\alpha}, \boldsymbol{\gamma}) > r\right] \leq e^{-\left(\frac{r}{32} - 1\right)s},$$ where $\mathbf{1}_{\mathcal{S}}$ is the $K$-length vector with the $k^{th}$ component equal to $1$ if $k \in \mathcal{S}$ and $0$ otherwise.
\end{lemma}
\begin{proof}
Since $\{\Theta_{kj}\}_{k\neq j}$ are  i.i.d.\ random variables with uniform distribution in $[-\pi, \pi)$, $\{\beta_{kj}\}_{j \neq k}$ are   i.i.d.\ with mean $\frac{1}{2}$. It follows that $\{R_k(\mathbf{1}_{\mathcal{S}}, \boldsymbol{\alpha}, \boldsymbol{\gamma})\}_{k \in \mathcal{S}}$ are  independent random variables. Using the Chernoff bound for independent random variables, we have
\begin{align}
\label{eqn:chernoff-outer}
\mathbb{P}\left[R_{sum}(\mathbf{1}_{\mathcal{S}}, \boldsymbol{\alpha}, \boldsymbol{\gamma}) > r\right] \leq e^{-\lambda r}\prod_{k \in \mathcal{S}}\mathbb{E}\left[e^{\lambda R_k(\mathbf{1}_{\mathcal{S}}, \boldsymbol{\alpha}, \boldsymbol{\gamma})}\right]
\end{align}
 for every $\lambda > 0.$ To find the expected value in the above expression, we once again use the Chernoff bound, but this time a version that holds for bounded  i.i.d.\ random variables (\cite[Theorem 4.5]{mitzenmacher-upfal05probability}). This bound can be stated as follows: if $\{B_j\}_{j=1}^n \in [0,1]$ are  i.i.d.\ with mean $\mathbb{E}B$, then 
$$\mathbb{P}\left[\sum_{j=1}^n B_j < (1-\delta)n\mathbb{E}B\right] \leq e^{-\frac{\delta^2}{2}n\mathbb{E}B}$$ for any $\delta \in (0, 1).$
Now, for any $k \in \mathcal{S},$ applying the above Chernoff bound to the  i.i.d.\ random variables $\{\beta_{kj}\}_{j \in \mathcal{S}, j \neq k} \in [0,1]$ with mean $\frac{1}{2}$ gives
\begin{align*}
\mathbb{P}\left[\sum_{j \in \mathcal{S}, j \neq k} \beta_{kj} < \frac{s-1}{4}\right] & \leq e^{-\frac{s-1}{16}}.
\end{align*}
From this,
\begin{align*}
& \mathbb{E}\left[e^{\lambda R_k(\mathbf{1}_{\mathcal{S}}, \boldsymbol{\alpha}, \boldsymbol{\gamma})}\right] \\
& = \mathbb{E}\left[\left( 1 + \frac{\beta_{kk}}{\sum_{j \in \mathcal{S}\backslash\{k\}} \beta_{kj} + \frac{1}{2}} \right)^\lambda \right]	\\
& \leq \left( 1 + \frac{1}{\frac{s-1}{4} + \frac{1}{2}} \right)^\lambda  + \left( 1 + \frac{1}{0 + \frac{1}{2}} \right)^\lambda e^{-\frac{s-1}{16}}	\\
& \leq e^{\frac{4\lambda}{s+1}} + 3^\lambda e^{-\frac{s-1}{16}}.
\end{align*}
With $\lambda = s/32,$ this yields $$\mathbb{E}\left[e^{\frac{s}{32} R_k(\mathbf{1}_{\mathcal{S}}, \boldsymbol{\alpha}, \boldsymbol{\gamma})}\right] \leq e, \; \forall s \geq 0.$$ Applying this inequality in \eqref{eqn:chernoff-outer} shows that
\begin{align*}
\mathbb{P}\left[R_{sum}(\mathbf{1}_{\mathcal{S}}, \boldsymbol{\alpha}, \boldsymbol{\gamma}) > r\right] \leq e^{-\left(\frac{r}{32} - 1\right)s}.
\end{align*}
\end{proof}

The next lemma shows that for any channel and for any subset of scheduled users, the difference between sum-rates obtained by two sets of transmit and receive directions that are \emph{close} is bounded by a constant.
\begin{lemma}
\label{lem:continuity}
For any scheduled subset of users, $\mathcal{S} \subset [K]$ with $ \lvert \mathcal{S} \rvert  = s,$ and for any $(\boldsymbol{\alpha}, \boldsymbol{\gamma}), (\boldsymbol{\alpha}', \boldsymbol{\gamma}') \in [-\pi, \pi)^{2K}$ such that $$\max_{k \in \mathcal{S}}(\max ( \lvert \alpha_k - \alpha_k' \rvert ,  \lvert \gamma_k - \gamma_k' \rvert )) \leq \frac{1}{2s^2},$$ the difference between the sum-rates satisfies $$ \lvert R_{sum}(\mathbf{1}_{\mathcal{S}}, \boldsymbol{\alpha}, \boldsymbol{\gamma}) - R_{sum}(\mathbf{1}_{\mathcal{S}}, \boldsymbol{\alpha}', \boldsymbol{\gamma}') \rvert  \leq 4.$$
\end{lemma}
\begin{proof}
For a fixed $\mathcal{S} \subset [K]$, let $f: [-\pi, \pi)^{2K} \rightarrow \mathbb{R}$ be defined as
\begin{align*}
f(\boldsymbol{\alpha}, \boldsymbol{\gamma}) := R_{sum}(\mathbf{1}_{\mathcal{S}}, \boldsymbol{\alpha}, \boldsymbol{\gamma}).
\end{align*}
Since $f$ is a differentiable function, by the Mean-Value Theorem, there exists $v \in [0,1]$ such that
\begin{align}
 \label{eqn:mean-value-thm}
 & f(\boldsymbol{\alpha}, \boldsymbol{\gamma}) - f(\boldsymbol{\alpha}', \boldsymbol{\gamma}')	\nonumber	\\
 & = \nabla f(v\boldsymbol{\alpha}+(1-v)\boldsymbol{\alpha}', v\boldsymbol{\gamma}+(1-v)\boldsymbol{\gamma}')^\intercal (\boldsymbol{\alpha} - \boldsymbol{\alpha}', \boldsymbol{\gamma} - \boldsymbol{\gamma}').
\end{align}
For any $k \notin \mathcal{S},$ $\frac{\partial f}{\partial \alpha_k}, \frac{\partial f}{\partial \gamma_k} = 0.$ 
For any $k \in \mathcal{S},$ it is fairly straightforward to show that $ \lvert \frac{\partial f}{\partial \alpha_k} \rvert  \leq 4s$ and $ \lvert \frac{\partial f}{\partial \gamma_k} \rvert  \leq 4s.$ Substituting these inequalities in \eqref{eqn:mean-value-thm}, we get
\begin{align*}
 \lvert f(\boldsymbol{\alpha}, \boldsymbol{\gamma}) - f(\boldsymbol{\alpha}', \boldsymbol{\gamma}') \rvert
   & \leq 4s\sum_{k \in \mathcal{S}}\left( \lvert \alpha_k - \alpha'_k \rvert  +  \lvert \gamma_k - \gamma'_k \rvert \right).
\end{align*}
For any $(\boldsymbol{\alpha}, \boldsymbol{\gamma}), (\boldsymbol{\alpha}', \boldsymbol{\gamma}') \in [-\pi, \pi)^{2K}$ such that $$\max_{k \in \mathcal{S}}\left(\max ( \lvert \alpha_k - \alpha_k' \rvert ,  \lvert \gamma_k - \gamma_k' \rvert )\right) \leq \frac{1}{2s^2},$$ we have
\begin{align*}
\sum_{k \in \mathcal{S}}\left( \lvert \alpha_k - \alpha'_k \rvert  +  \lvert \gamma_k - \gamma'_k \rvert \right) \leq \frac{1}{s}.
\end{align*}
Therefore, the difference between the sum-rates satisfies
\begin{align*}
\lvert R_{sum}(\mathbf{1}_{\mathcal{S}}, \boldsymbol{\alpha}, \boldsymbol{\gamma})
  - R_{sum}(\mathbf{1}_{\mathcal{S}}, \boldsymbol{\alpha}', \boldsymbol{\gamma}') \rvert
 & \leq 4.
\end{align*}
\end{proof}
We now prove Theorem~\ref{thm:upper-bound}. 
\begin{proof}[Proof of Theorem~\ref{thm:upper-bound}]
We show that within the subclass of single-symbol phase alignment schemes that schedule users transmitting at maximum power, it is not possible to achieve a sum-rate better than $O(\ln K) \; w.h.p.$ i.e.,
\begin{align}
\label{eqn:schedule-upper-bound}
\max_{\substack{\mathbf{P} \in \{0,1\}^K\\ (\boldsymbol{\alpha}, \boldsymbol{\gamma}) \in [-\pi, \pi)^{2K}}}
 R_{sum}(\mathbf{P}, \boldsymbol{\alpha}, \boldsymbol{\gamma}) & = O(\ln K) \; w.h.p.
\end{align}
The general result then follows from Lemma~\ref{lem:power-ctrl=schedule}, which shows that it is sufficient to consider this subclass of scheduling schemes for asymptotic analysis.

To prove \eqref{eqn:schedule-upper-bound}, we use Lemmas~\ref{lem:fixed-parameters} and \ref{lem:continuity} to obtain a probabilistic upper bound on the achievable sum-rate for a fixed scheduled set when the transmit and receive directions are restricted within in a small set that we call a \emph{T-set}. To extend this upper bound to the set of all possible transmit and receive directions, we use union bound over T-sets that cover $[-\pi, \pi)^{2K}.$ 

The T-set $T(\mathcal{S}, \boldsymbol{\alpha}, \boldsymbol{\gamma}),$ parametrized by a subset of scheduled users $\mathcal{S} \subseteq [K]$ and a set of transmit and receiver directions $(\boldsymbol{\alpha}, \boldsymbol{\gamma})\in [-\pi, \pi)^{2K},$ is a subset of $[-\pi, \pi)^{2K}$ defined as follows: $(\boldsymbol{\alpha}', \boldsymbol{\gamma}') \in T(\mathcal{S}, \boldsymbol{\alpha}, \boldsymbol{\gamma})$ if
$$\max_{k \in \mathcal{S}}(\max ( \lvert \alpha_k - \alpha_k' \rvert ,  \lvert \gamma_k - \gamma_k' \rvert )) \leq \frac{1}{2 \lvert \mathcal{S} \rvert ^2}.$$
In words, $T(\mathcal{S}, \boldsymbol{\alpha}, \boldsymbol{\gamma})$ is a cylinder set in the space of all possible transmit and receive directions $(\boldsymbol{\alpha}', \boldsymbol{\gamma}') \in [-\pi, \pi)^{2K}$ with the set of transmit and receive directions for user set $\mathcal{S}$ restricted to the $ \lvert \mathcal{S} \rvert $-dimensional hypercube of length $\frac{1}{ \lvert \mathcal{S} \rvert ^2}$ with $(\boldsymbol{\alpha}, \boldsymbol{\gamma})$ as the center. 

Now consider a fixed user set $\mathcal{S}$ with $ \lvert \mathcal{S} \rvert  = s$. For any fixed $(\boldsymbol{\alpha}, \boldsymbol{\gamma}) \in [-\pi, \pi)^{2K},$ combining Lemmas~\ref{lem:fixed-parameters} and \ref{lem:continuity} gives
\begin{align}
\label{eqn:T-set}
& \mathbb{P}\left[\max_{(\boldsymbol{\alpha}', \boldsymbol{\gamma}') \in T(\mathcal{S}, \boldsymbol{\alpha}, \boldsymbol{\gamma})}R_{sum}(\mathbf{1}_{\mathcal{S}}, \boldsymbol{\alpha}', \boldsymbol{\gamma}') > r + 4\right]	\nonumber	\\
& \leq \mathbb{P}\left[R_{sum}(\mathbf{1}_{\mathcal{S}}, \boldsymbol{\alpha}, \boldsymbol{\gamma}) > r \right]	\nonumber	\\
  & \leq e^{-\left(\frac{r}{32} - 1\right)s}.
\end{align}
We can now cover $[-\pi, \pi)^{2K}$, the space of all possible transmit-receive directions, by the collection of $T$-sets corresponding to  $\mathcal{S}$ and all $(\boldsymbol{\alpha}, \boldsymbol{\gamma}) \in V(\mathcal{S})$. Here, $V(\mathcal{S})$ is a subset of $[-\pi, \pi)^{2K}$ defined as follows:
$(\boldsymbol{\alpha}, \boldsymbol{\gamma}) \in V(\mathcal{S})$ if $\alpha_k, \gamma_k = 0 \, \forall k \notin \mathcal{S}$ and $$\alpha_k, \gamma_k \in \left\{\frac{n}{s^2} : n \in \mathbb{Z},  \lvert n \rvert  < \pi s^2 + \frac{1}{2}\right\} \, \forall k \in \mathcal{S}.$$
$V(\mathcal{S})$ contains transmit-receive directions that are spread apart by $\frac{1}{s^2}$ across dimensions corresponding to the scheduled user set $\mathcal{S}$. Note that $[-\pi, \pi)^{2K} \subseteq \bigcup_{(\boldsymbol{\alpha}, \boldsymbol{\gamma}) \in V(\mathcal{S})} T(\mathcal{S}, \boldsymbol{\alpha}, \boldsymbol{\gamma})$
 and $ \lvert V(\mathcal{S}) \rvert  \leq (2\pi s^2 + 2)^s.$ Taking a union bound over $T$-sets corresponding to $\mathcal{S}$ and $(\boldsymbol{\alpha}, \boldsymbol{\gamma})\in V(\mathcal{S}),$
\begin{align}
\label{eqn:upper-bound-fixed-S}
\begin{split}
& \mathbb{P}\left[\max_{(\boldsymbol{\alpha}, \boldsymbol{\gamma}) \in [-\pi, \pi)^{2K}} R_{sum}(\mathbf{1}_{\mathcal{S}}, \boldsymbol{\alpha}, \boldsymbol{\gamma}) > r + 4\right]	\\
& \leq  \mathbb{P}\left[\bigcup_{(\boldsymbol{\alpha}, \boldsymbol{\gamma}) \in V(\mathcal{S})}\max_{(\boldsymbol{\alpha}', \boldsymbol{\gamma}') \in T(\mathcal{S}, \boldsymbol{\alpha}, \boldsymbol{\gamma})}R_{sum}(\mathbf{1}_{\mathcal{S}}, \boldsymbol{\alpha}', \boldsymbol{\gamma}') > r + 4\right] \\
& \leq \sum_{(\boldsymbol{\alpha}, \boldsymbol{\gamma}) \in V(\mathcal{S})} \mathbb{P}\left[\max_{(\boldsymbol{\alpha}', \boldsymbol{\gamma}') \in T(\mathcal{S}, \boldsymbol{\alpha}, \boldsymbol{\gamma})}R_{sum}(\mathbf{1}_{\mathcal{S}}, \boldsymbol{\alpha}', \boldsymbol{\gamma}') > r + 4\right] \\
& \leq \sum_{(\boldsymbol{\alpha}, \boldsymbol{\gamma}) \in V(\mathcal{S})}  e^{-\left(\frac{r}{32} - 1\right)s} \quad (\text{from } \eqref{eqn:T-set})	\\
& \leq (2\pi s^2 + 2)^s e^{-\left(\frac{r}{32} - 1\right)s}.
\end{split}
\end{align}
This inequality gives a probabilistic upper bound on the achievable sum-rate for a fixed set of scheduled users.
We can now use \eqref{eqn:upper-bound-fixed-S} to prove \eqref{eqn:schedule-upper-bound} by taking a union bound over all possible scheduling sets. Substituting $r = C \ln K$ in \eqref{eqn:upper-bound-fixed-S} with $C = 6 \times 32,$
\begin{align*}
& \mathbb{P}\left[\max_{\substack{\mathcal{S} \subset [K]\\ (\boldsymbol{\alpha}, \boldsymbol{\gamma}) \in [-\pi, \pi)^{2K}}} R_{sum}(\mathbf{1}_{\mathcal{S}}, \boldsymbol{\alpha}, \boldsymbol{\gamma}) > C \ln K + 4\right]	\\
& \leq \sum_{\mathcal{S} \subset [K]} \mathbb{P}\left[\max_{(\boldsymbol{\alpha}, \boldsymbol{\gamma}) \in [-\pi, \pi)^{2K}} R_{sum}(\mathbf{1}_{\mathcal{S}}, \boldsymbol{\alpha}, \boldsymbol{\gamma}) > C \ln K + 4\right]	\\
& \leq \sum_{s=1}^K \binom K s (2\pi s^2 + 2)^s e^{-\left(\frac{C \ln K}{32} - 1\right)s} \quad (\text{from  \eqref{eqn:upper-bound-fixed-S}})	\\
& \leq \sum_{s=1}^K K^s e^{s \ln(2\pi s^2 + 2) -s\left(\frac{C \ln K}{32} - 1\right)}	\\
& = \sum_{s=1}^K e^{s\left(\ln K + \ln(2\pi s^2 + 2) + 1 - \frac{C \ln K}{32}\right)}	\\
& \leq \sum_{s=1}^K e^{s\left(4\ln K - \frac{C \ln K}{32}\right)} \quad (\text{for $K$ large enough})	\\
& = \sum_{s=1}^K e^{-2s\ln K } \quad (\text{since } C = 6 \times 32)	\\
& \leq K e^{-2\ln K}	\\
& = \frac{1}{K},
\end{align*}
which proves \eqref{eqn:schedule-upper-bound} and thus Theorem~\ref{thm:upper-bound}.
\end{proof}

\end{document}